\documentclass[intlimits]{amsart}

\usepackage{amsmath,amssymb,amsthm}
\usepackage{graphicx}

\newtheorem{lemma}{Lemma}

\newtheorem{theorem}[lemma]{Theorem}
\newtheorem{cor}[lemma]{Corollary}

\newtheorem{lemmaA}{Lemma}[section]
\newtheorem{corA}[lemmaA]{Corollary}
\newtheorem{theoremA}[lemmaA]{Theorem}

\theoremstyle{remark}
\newtheorem{rem}{Remark}
\newtheorem*{ack}{Acknowledgements}

\newcommand{\Pel}{P_{el,j}}
\newcommand{\Pelbar}{\overline{P}_{el,j}}
\newcommand{\PO}{P}
\newcommand{\PObar}{\overline{P}}
\newcommand{\Hel}{H_{el}}
\newcommand{\chiH}{\chi_{H_f\leq \rho_0}}

\newcommand{\G}[4]{{w_{#1,#2}^{(#3)}(#4,\mu)}}

\newcommand{\cF}[1][z]{\mathcal{F}_{\PO(\theta)}(H_g(\theta)-#1)}

\newcommand{\cA}{\mathcal{A}(\delta, \epsilon)}
\newcommand{\cO}{\mathcal{O}}
\newcommand{\cU}{\mathcal{U}}

\newcommand{\idf}{\mathbf{1}_f}
\newcommand{\id}{\mathbf{1}}
\newcommand{\idel}{\mathbf{1}_{el}}

\newcommand{\hilbert}{\mathcal{H}}

\newcommand{\fock}{\mathcal{F}}
\newcommand{\asym}{\mathcal{A}_N}
\newcommand{\sym}{\mathcal{S}_N}
\newcommand{\C}{\mathbb{C}}
\newcommand{\Z}{\mathbb{Z}}
\newcommand{\R}{\mathbb{R}}
\newcommand{\N}{\mathbb{N}}

\renewcommand{\Im}{{\rm Im}\,}
\renewcommand{\Re}{{\rm Re}\,}

\DeclareMathOperator{\NumRan}{NumRan}
 \DeclareMathOperator{\ran}{Ran}
\DeclareMathOperator{\dom}{Dom} \DeclareMathOperator{\dist}{dist}

\begin{document}
    \title{On the Lifetime of Quasi-Stationary States
    in Non-Relativistic QED}
    \author[D. Hasler]{David Hasler}
   \address{David Hasler\\ Department of Mathematics\\
    P. O. Box 400137\\ University of Virginia\\ Charlottesville, VA
    22904-4137\\United States of America }
    \email{iwh@virginia.edu}

     \author[I. Herbst]{Ira Herbst}
     \address{Ira Herbst\\ Department of Mathematics\\
    P. O. Box 400137\\ University of Virginia\\ Charlottesville, VA
    22904-4137\\United States of America}
    \email{dh8ud@cms.mail.virginia.edu}

      \author[M. Huber]{Matthias Huber}
        \address{Matthias Huber\\ Mathematisches Institut\\
        Ludwig-Maximilians-Uni\-ver\-si\-t\"{a}t M\"{u}n\-chen\\ Theresienstra{\ss}e 39\\ 80333
        M\"{u}nchen\\ Germany}
        \email{mhuber@math.lmu.de}

    \date{September 24, 2007}

    \begin{abstract}
        We consider resonances in the Pauli-Fierz model of non-re\-la\-ti\-vis\-tic
        QED.  We use and slightly modify  the analysis
        developed by Bach, Fr\"{o}hlich and Sigal
        \cite{Bachetal1998Q,Bachetal1999S} to obtain an upper and \emph{lower} bound on the lifetime of quasi-stationary
        states.
    \end{abstract}
    \maketitle

    \section{Introduction and Main Result}\label{sec:1}
        Spectral properties of models of non-relativistic QED were
        investigated by Bach, Fr\"{o}hlich, Sigal and Soffer
        \cite{Bachetal1995,Bachetal1998Q,Bachetal1999S,Bachetal1999P}
        and by many others. Bach, Fr\"{o}hlich, and Sigal \cite{Bachetal1999S} proved,
        among other things, an upper bound on the lifetime of
        quasi-stationary states.

        We show an upper and  \emph{lower} bound on the
        lifetime of quasi-stationary states. We heavily rely  on the
        analysis developed in \cite{Bachetal1998Q,Bachetal1999S}, but choose
        a different contour of integration and make use of an
        additional cancellation of terms. Moreover, we neither
        require a non-degeneracy assumption nor a spectral
        cutoff. However, we do not provide time dependent estimates on the remainder
        term and there are no photons in our quasi-stationary state.
        Estimates similar to ours
        were obtained before by different authors for other models,
        see e.g. \cite{JaksicPillet1995,King1994}.

        In order to be self-contained, we give all the necessary
        definitions for the model considered. For details, we refer
        the reader to \cite{Bachetal1999S}.
        We consider an atom in interaction with the second quantized
        electromagnetic field. The Hilbert space of the system is
        given by
        $$\hilbert:=\hilbert_{el}\otimes\fock,$$
        where
        $$\hilbert_{el}:=\asym L^2[(\R^3\times \Z_2)]^N$$
        is the Hilbert space of $N$ electrons with spin,
        and where
        $$\fock:=\bigoplus_{N=0}^\infty \sym L^2[(\R^3\times
        \Z_2)]^N$$
        is the Fock space (with vacuum $\Omega$) of the quantized electromagnetic field,
        allowing two transverse polarizations of the photon. $\asym$
        and $\sym$ are the projections onto the subspaces of
        functions anti-symmetric and symmetric, respectively, under
        a  permutation of variables. Strictly speaking, we would
        have to take the physical units into account in the
        definition of these spaces. However, we refrain from doing
        so in order not to complicate the notation.
        The operator
         $$\Hel':=-\frac{\hbar^2}{2m}\Delta_{3N}
         +\frac{\mathfrak{e}^2}{4\pi\epsilon_0}
         \left[\sum_{j=1}^N\frac{-\mathcal{Z}}{|x_j|}+\sum_{1\leq i< j\leq N}\frac{1}{|x_i-x_j|}\right]$$
        describes the electrons, and the operator for the total system
        is
         \begin{multline*}H_g':=\frac{1}{2m}\sum_{j=1}^N:[\sigma_j\cdot (-i\hbar\nabla_{x_j}-\mathfrak{e}A'_{\kappa'}(x_j))]^2:
         +H_f'\\+\frac{\mathfrak{e}^2}{4\pi\epsilon_0}\left[\sum_{j=1}^N\frac{-\mathcal{Z}}{|x_j|}+\sum_{1\leq i< j\leq
         N}\frac{1}{|x_i-x_j|}\right],\end{multline*}
        where $-\mathfrak{e}\mathcal{Z}$ is the charge of the nucleus, $\mathfrak{e}<0$ the
        charge of the electron, $\mathfrak{c}$ the velocity of light, $\hbar$ is
         Planck's constant, $\epsilon_0$ is the permittivity of the vacuum, $m$ the mass of the
        electron, $\sigma_j$ is the Pauli matrix for the jth electron, and $:\cdots:$ denotes normal ordering.
        The kinetic energy of the photons is
        $$H_f':=\hbar\mathfrak{c} \sum_{\mu=1,2}\int\nolimits_{k\in \R^3} dk  |k|a'^*_\mu(k) a'_\mu(k),$$
        where the $a'^*_\mu(k)$ and $a'_\mu(k)$ are the usual
        creation and annihilation operators.
        The second quantized electromagnetic field is
        $A'_{\kappa'}(x):=A'_{\kappa'}(x)_++A'_{\kappa'}(x)_-$,
        where
        \begin{equation*}
            A'_{\kappa'}(x)_+:=\sum_{\mu=1,2}\int dk
            \kappa'(|k|)\sqrt{\frac{\hbar}{2\epsilon_0\mathfrak{c}|k|(2\pi)^3}}\varepsilon'_\mu(k)e^{-ik\cdot x}a'^*_\mu(k).
        \end{equation*}
        and
        \begin{equation*}
            A'_{\kappa'}(x)_-:=\sum_{\mu=1,2}\int dk
            \kappa'(|k|)\sqrt{\frac{\hbar}{2\epsilon_0\mathfrak{c}|k|(2\pi)^3}}\varepsilon'_\mu(k)e^{ik\cdot x}a'_\mu(k).
        \end{equation*}
        Here $\varepsilon'_\mu(k)$, $\mu=1,2$, are the
        polarization vectors of the photon, depending only on the direction of $k$.
        Let us note that we use SI units here; for details about
        these operators, we refer the reader to
        \cite{CohenTannoudjietal1992,CohenTannoudjietal2004}.
        We set $a_0:=\alpha^{-1}(\frac{\hbar}{ m\mathfrak{c}})$
        (Bohr radius), $\zeta:=\frac{a_0}{2}$ and
        $\xi^{-1}:=\frac{2\alpha}{a_0}$. Moreover, $\kappa'(r):=\kappa(r \xi)$ is a cutoff
        function depending on the fine structure constant $\alpha=\frac{\mathfrak{e}^2}{4\pi\epsilon_0\hbar \mathfrak{c}}$.
        $\kappa$ is a function, which is positive on $[0,\infty)$,
        satisfies $\kappa(r)\rightarrow 1$ as $r\rightarrow 0$, and has an analytic continuation to a cone around the positive real axis which
        is bounded and decays faster than any inverse polynomial, e.g.,
        $\kappa(r):=e^{-r^4}$.
        Following \cite{Bachetal1998Q}, we scale the operator with
        the transformation
         $x_j\rightarrow
        \zeta x_j$ and $k\rightarrow \xi^{-1} k$. We denote the
        corresponding unitary transformation by $U$.
        After this transformation the electron positions are measured in units of
        $\frac{1}{2}a_0$, photon
        wave vectors in units of $\frac{2\alpha}{a_0}$,
        and energies in units of $4{\rm Ry}$, where the Rydberg is ${\rm Ry}:=\frac{\alpha^2mc^2}{2}$.
        The creation an annihilation operators transform as
        $$U a'_\mu(k) U^{-1}=\xi^{3/2}a_\mu(\xi k),\quad
        U a'^*_\mu(k) U^{-1}=\xi^{3/2}a^*_\mu(\xi k). $$
        Moreover, we set
        $$\varepsilon_\mu(k):=\varepsilon'_\mu(\xi^{-1}k),\quad
        \mu=1,2.$$
        Accordingly, we obtain
        $$U H_g' U^{-1}=2\alpha^2(mc^2) H_g,$$
        with  $H_g:=H_0+W_g$ and
        $H_0:=\Hel\otimes \idf +\idel \otimes H_f$, where
        $$\Hel:=-\Delta_{3N}+\sum_{j=1}^N\frac{-\mathcal{Z}}{|x_j|}+\sum_{1\leq i< j\leq N}\frac{1}{|x_i-x_j|}.$$
        Here
        $$H_f:= \sum_{\mu=1,2}\int\nolimits_{k\in \R^3} dk  |k|a^*_\mu(k) a_\mu(k)$$
        and
        the
        interaction is given by
        \begin{multline*}
            W_g:=\sum_{j=1}^N\{2\alpha^{3/2} A_\kappa(\alpha
            x_j)\cdot(-i\nabla_{x_j})+\alpha^3 :A_\kappa^2(\alpha
            x_j):\\+\alpha^{5/2}\sigma_j\cdot (\nabla\times
            A_\kappa)(\alpha x_j)\},
        \end{multline*}
        where
        the second quantized electromagnetic field is
        $A_{\kappa}(x):=A_{\kappa}(x)_++A_{\kappa}(x)_-$
        with
        \begin{equation*}
            A_{\kappa}(x)_+:=\sum_{\mu=1,2}\int \frac{dk
            \kappa(|k|)}{\sqrt{4\pi^2|k|}}\varepsilon_\mu(k)e^{-ik\cdot x}a^*_\mu(k).
        \end{equation*}
        and
        \begin{equation*}
            A_{\kappa}(x)_-:=\sum_{\mu=1,2}\int \frac{dk
            \kappa(|k|)}{\sqrt{4\pi^2|k|}}\varepsilon_\mu(k)e^{ik\cdot x}a_\mu(k).
        \end{equation*}

        As in \cite{Bachetal1999S},  we set $g:=\alpha^{3/2}$. Henceforth, we let the coupling constant $g:=\alpha^{3/2}>0$ be the perturbation
        parameter.
        We assume that the spectrum of $\Hel$ has the structure
        $$\sigma(\Hel)=\{E_0, E_1, \ldots \}\cup[\Sigma,\infty),$$
        where $\Sigma:=\inf \sigma_{ess}(\Hel)$ and $E_0<E_1<\ldots$
        are (at least two) eigenvalues (possibly)
        accumulating at $\Sigma$.
        In the following, we will look at one (fixed) eigenvalue $E_j$ of $\Hel$  with $j\geq1$.
        For $0<\epsilon<1/3$ we set $\rho_0:=g^{2-2\epsilon}$, and $\cA:=[E_j-\delta/2,
        E_j+\delta/2]+i[-g^{2-\epsilon}, \infty)$, where
        $\delta:=\dist(E_j, \sigma(\Hel)\setminus\{E_j\})>0.$
        We define  the operators
        \begin{multline}\label{Eq:OpDilatDef}
        \Hel(\theta):=\cU_{el}(\theta)\Hel\cU_{el}(\theta)^{-1}, \,
        H_g(\theta):=\cU(\theta)H_g\cU(\theta)^{-1},\\W_g(\theta):=\cU(\theta)W_g\cU(\theta)^{-1} \end{multline}
         for real
        $\theta$,
        where $\cU(\theta)$ is the unitary group associated to the
        generator of dilations. It is defined in such a way that the space coordinates
        of the electrons are dilated as $x_j\mapsto e^{\theta} x_j$
        and the momentum coordinates of the photons as $k\mapsto e^{-\theta} k$.
        It can be shown \cite[Corollary 1.3, Corollary 1.4]{Bachetal1999S} that the operators defined
         in equation \eqref{Eq:OpDilatDef} are analytic families
         for $|\theta|\leq \theta_0$ for some $\theta_0>0$.
         We introduce the convention
        $\theta:=i\vartheta$ with $\vartheta>0$. Moreover,
        $\cU_{el}(\theta)$ is the above dilation acting on the
        electronic space only.

        We define (with $r>0$ small enough) $P_{el,i}(\theta):=-(2\pi i)^{-1}\int_{|E_i-z|=r}(\Hel(\theta)-z)^{-1}dz$ to be the projection onto the
        eigenspace corresponding to the eigenvalue $E_i$ of $\Hel(\theta)$ and set
        $\overline{P}_{el,i}(\theta):=1 -P_{el,i}(\theta)$. Furthermore, we define
        $\PO(\theta):=\Pel(\theta)\otimes \chiH$ and
        $\PObar(\theta):=1-\PO(\theta)$. We abbreviate
        $P_{el,i}:=P_{el,i}(0)$.

        Note that if we consider operators of
        the form $PAP$, where $A$ is a closed operator and
        $P$ a projection with $\dom A\subset \ran P$, then our notation does not
        distinguish between the operators $PAP$ and $PAP|_{\ran
        P}$. It will be clear from the context, how the
        symbol $PAP$ is to be understood.

        Following \cite{Bachetal1999S}, we make crucial use of the
        Feshbach operator
        \begin{multline}\label{Eq:Fesh}
            \cF:=\PO(\theta) (H_g(\theta)-z)\PO(\theta)\\-
            \PO(\theta) W_g(\theta) \PObar(\theta) [\PObar(\theta)(H_g(\theta)-z)\PObar(\theta)]^{-1}\PObar(\theta)
            W_g(\theta)\PO(\theta).
        \end{multline}
        For the convenience of the reader, we summarize its
        most
        important properties including its existence in Appendix A. For details, we refer
        the reader to \cite[Section IV]{Bachetal1998Q} and
        \cite{Bachetal1999S}.
         It was shown in \cite{Bachetal1998Q,Bachetal1999S} that the Feshbach operator can be
        approximated in a sense to be shown  using the  operators
        \begin{multline}\label{Eq:ZDef}
            \tilde Z_j^{od}(\alpha):=\lim_{\epsilon\downarrow 0}\sum_{\mu=1,2}\int\nolimits_{k\in \R^3} dk \Pel \G{0}{1}{0}{k}\\\times\Pelbar
              [\Pelbar \Hel-
            E_j+|k|-i\epsilon]^{-1}\Pelbar \G{1}{0}{0}{k}\Pel
        \end{multline}
        and
         \begin{equation}\label{Eq:ZDef1}
            \tilde Z_j^{d}(\alpha):=\sum_{\mu=1,2}\int\nolimits_{k\in \R^3} \frac{dk}{|k|} \Pel \G{0}{1}{0}{k}  \Pel
            \G{1}{0}{0}{k}\Pel.
        \end{equation}
        Here the coupling functions $\G{0}{1}{\theta}{k}$ and $\G{1}{0}{\theta}{k}$ will be needed later with $\theta\neq 0$.
        Denoting the momentum of the jth electron by $p_j$, they are
        \begin{equation}\label{Eq:CoupDef}
            \G{0}{1}{\theta}{k}:=\G{1}{0}{\bar\theta}{k}^*:=\sum_{j=1}^N\{2e^{-\theta} G^{(\theta)}_{x_j}(k,\mu)\cdot
            p_j +\sigma_j\cdot B^{(\theta)}_{x_j}(k,\mu)\},
        \end{equation}
        where
        \begin{equation}\label{Eq:CoupDef1}
            G^{(\theta)}_x(k,\mu):=\frac{e^{-\theta}
            \kappa(e^{-\theta}|k|)}{\sqrt{4\pi^2
            |k|}}e^{i\alpha k\cdot x}\epsilon_\mu(k)
        \end{equation}
        and
        \begin{equation}\label{Eq:CoupDef2}
            B^{(\theta)}_x(k,\mu):=\frac{\alpha e^{-2\theta} \kappa(e^{-\theta}|k|)}{i
            \sqrt{4\pi^2
            |k|}}e^{i\alpha k\cdot x}(k\times
            \epsilon_\mu(k)).
        \end{equation}
        We set
        \begin{multline}\label{Eq:ZDef2}
            \tilde Z(\alpha):=Z_j^{d}(\alpha)+Z_j^{od}(\alpha),\,\,\, \tilde Z(\alpha, \theta):=\cU_{el}(\theta)\tilde Z(\alpha)\cU_{el}(\theta)^{-1}
            ,\\ Z(\theta):=\tilde Z(0, \theta)\text{, and } Z:=\tilde Z(0, 0).
        \end{multline}
        We consider the Feshbach operator $\cF$ as an
                operator on $\ran \PO(\theta)$. Similarly, we consider $\tilde
                Z(\alpha):=Z_j^{d}(\alpha)+Z_j^{od}(\alpha)$ and  $\tilde Z(\alpha,
                \theta)$ as operators on $\ran \Pel$ and  $\ran
                \Pel(\theta)$ respectively.

        We are now able to formulate our main result. It will be
        proven in Section \ref{sec:2}.
        \begin{theorem}\label{Thm:TimeDecay} Let $0<\epsilon<1/3$
        and $g$ small enough.
            Let $\phi_1$ and $\phi_2$ be  normalized eigenvectors of $\Hel$ with
            eigenvalue $E_j$  and $\Phi_i:=\phi_i\otimes \Omega$.
            Assume moreover that
         the imaginary part $\Im Z:=\tfrac{1}{2i}(Z-Z^*)$ of $Z$ is strictly
            positive on $\ran P_{el,j}$.
            Then, in terms of a dimensionless time parameter $s\geq 0$, $$\langle \Phi_1,
            e^{-isH_g}\Phi_2\rangle=\langle \phi_1, e^{-is(E_j-g^2 Z)} \phi_2\rangle+b(g, s),$$
            where $|b(g, s)|\leq C g^\epsilon$ for some $C\geq 0$.
        \end{theorem}
        The theorem has the following immediate Corollary:
        \begin{cor}
            Under the assumptions of Theorem \ref{Thm:TimeDecay},  if $0<\tau:=g^2 s$ is kept fixed, if
            $\phi:=\phi_1=\phi_2$ are eigenvectors of $Z$
            with eigenvalue $\Gamma$, and if $\Phi:=\phi\otimes \Omega$,
            then
            $$\lim_{g\downarrow 0} |\langle \Phi,
            e^{-isH_g}\Phi\rangle|=e^{-\tau \Im \Gamma}.$$
        \end{cor}
        We close the introductory section with the following
        remarks:
        \begin{rem}
            The theorem can be rewritten in terms of the
             original operators:
            Let $\phi'_1$ and $\phi'_2$ be  normalized eigenvectors of $\Hel'$ with
            eigenvalue $2\alpha^2mc^2E_j$  and $\Phi'_i:=\phi'_i\otimes \Omega$.
            Then $\langle \Phi'_1,
            e^{-it\hbar^{-1}H'_g}\Phi'_2\rangle=\langle \phi'_1, e^{-it\frac{2\alpha^2mc^2}{\hbar}(E_j-g^2 Z')} \phi'_2\rangle
            +\cO(\alpha ^{3\epsilon/2})=
\langle \phi_1, e^{-it\frac{2\alpha^2mc^2}{\hbar}(E_j-g^2 Z)}
\phi_2\rangle
             +\cO(\alpha ^{3\epsilon/2})$,   where $\phi_i\otimes \Omega = U [\phi'_i\otimes \Omega]$.

            Here
            \begin{multline}
                Z':=\frac{\hbar^2}{8\alpha^4m^3c^2}\Bigg [ \lim_{\epsilon\downarrow 0}\sum_{\mu=1,2}\int\nolimits_{k\in \R^3} dk
                \frac{\kappa '(|k|)^2}{4\pi^2|k|}
                \Pel'\epsilon_\mu(k)\cdot p'
                \Pelbar'\\\times
              [\Pelbar' \Hel'-
            2\alpha^2mc^2E_j+\hbar c|k|-i\epsilon]^{-1}\Pelbar'\epsilon_\mu(k)\cdot p' \Pel'\\
            + \sum_{\mu=1,2}\int\nolimits_{k\in \R^3} \frac{dk}{\hbar c |k|} \frac{ \kappa' (|k|)^2}{4\pi^2|k|}
             \Pel' \epsilon_\mu(k)\cdot p'  \Pel'
            \epsilon_\mu(k)\cdot p'\Pel'
             \Bigg ]
            \end{multline}
            $P_{el,i}'$ is the projection onto the eigenspace of
            $\Hel'$ belonging to the eigenvalue $2\alpha^2mc^2E_i$,
            $p'_j:=-i\hbar\nabla_{x_j}$ and $p':=\sum_{j=1}^np'_j$.
        \end{rem}

        \begin{rem}
            Note that the matrix $\tilde Z(\alpha)$
            depends on the fine structure constant $\alpha$, since the
            coupling functions defined in equations
            \eqref{Eq:CoupDef}, \eqref{Eq:CoupDef1}, and
            \eqref{Eq:CoupDef2} do. Thus, due to the exponential decay of
            the eigenfunctions of the electronic operator, $\tilde Z(\alpha)$ can
            be developed in a power series in $\alpha= g^{2/3}$. The zero
            order term corresponds to electric dipole
            (E1)
            transitions, the higher order terms to magnetic dipole
            transitions as well as to higher order electric and
            magnetic transitions. We have for some $C>0$
            \begin{equation}\label{Eq:ZAlphaZero}
                g^2\|\tilde Z(\alpha, \theta)-Z(\theta)\|\leq C g^{2+2/3}.
            \end{equation}
            It is easy to see that the imaginary part of
            $Z$ is (see also \cite[Formula (IV.19)]{Bachetal1998Q})
            \begin{multline}
                \Im Z=\pi
                \sum_{i=0}^{j-1}\sum_{\mu=1,2}\int\nolimits_{|\omega|=1}d\omega
                (E_i-E_j)^2\\\times \frac{4
                \kappa(E_j-E_i)^2}{4\pi^2
                (E_j-E_i)}\Pel[\epsilon_\mu(\omega)\cdot p]  P_{el,i}[\epsilon_\mu(\omega)\cdot p]
                \Pel=\\
                \frac{8}{3}
                \sum_{i=0}^{j-1}
                (E_j-E_i) \kappa(E_j-E_i)^2
            \Pel p  P_{el,i} p
            \Pel,
            \end{multline}
            In the last step we used the relationships $\sum_{\mu=1,2}
            (\epsilon_\mu(\omega))_{m}(\epsilon_\mu(\omega))_{n}=\delta_{m,n}-\omega_m
            \omega_n$ and $\int d\omega \omega_m
            \omega_n=\frac{4\pi\delta_{m,n}}{3}$, where $\delta_{m,n}$
            is the Kronecker symbol.
             Moreover, $p:=\sum_{j=1}^Np_j$ and the expression $\Pel p  P_{el,i} p
            \Pel$ indicates a Euclidean inner product. Analogously,
            we set $x:=\sum_{j=1}^Nx_j$.
             Using the commutation relation
            $$[x,\Hel]=2ip$$
            we find
            \begin{equation}\label{Eq:ImZX}
                \Im Z=
                \frac{2}{3}
                \sum_{i=0}^{j-1}
                (E_j-E_i)^3 \kappa(E_j-E_i)^2
            \Pel x  P_{el,i} x
            \Pel.
            \end{equation}
            We analyze equation \eqref{Eq:ImZX} for the case of
            a hydrogen atom in Appendix B. We show there that
            $\Im Z$ is indeed strictly positive unless $j=1$. If
            $j=1$, $\Im Z$ has a zero eigenvalue, since
              the $2s$ state of hydrogen cannot decay
            via electric dipole transitions. However, the $2p$ states
             can decay via an electric dipole transition. It
            would be interesting to prove time decay estimates also
            in the latter case.

         Note that the transition rate is
            proportional to $g^2\alpha^2\propto \alpha^5$, in accordance
            with physics textbooks (see e.g.
            \cite[Section
            59]{BetheSalpeter1957Q}).

        \end{rem}
         \begin{rem}\label{Rem:ZSim}
                The eigenvectors of $\Hel$ are analytic vectors for the
                generator of dilations, and therefore
                $\cU_{el}(\theta):\ker(\Hel(0)-E_j)\rightarrow \ker(\Hel(\theta)-E_j)$ is a (bounded and bounded
                invertible) mapping between finite dimensional vector
                spaces. (The latter is true for $|\Im \theta|<\pi/2$.) This implies that the matrices $Z(0)$ and
                $Z(\theta)$ are similar.  In particular,
                the
                bounded operators $[-g^2Z(0)\otimes \idf
                +e^{-\theta}\idel\otimes H_f]|_{\ran \PO(0)}$ and $[-g^2Z(\theta)\otimes \idf
                +e^{-\theta}\idel\otimes H_f]|_{\ran \PO(\theta)}$ are
                similar (cf. \cite[Section 3]{Bachetal1999S}).
                This fact will be used in the proof of Theorem \ref{Thm:TimeDecay} and in Section \ref{sec:3}.
        \end{rem}

    \section{Proof of the Main Result}\label{sec:2}
        In this section we prove our main result. The technical
        estimates needed in the proof are collected in a series of lemmas and deferred to Section
        \ref{sec:3}.
        For the proof we need the operator (see \cite[Formula
        (IV.67)]{Bachetal1998Q})
        \begin{multline}
            Q^{(\theta)}(z):=\sum_{\mu=1,2}\int\nolimits_{k\in \R^3} dk \PO(\theta) [\G{0}{1}{\theta}{k}\otimes \idf]\\\times  \left
            [\frac{\PObar(\theta)(|k|)}{\Hel(\theta)+e^{-i\vartheta}(H_f+|k|)-z}\right]
            [\G{1}{0}{\theta}{k}\otimes \idf] \PO(\theta),
        \end{multline}
        defined on $\ran \PO(\theta)$ and for $z \in \cA$. Here we used the definition
        $\overline{P}(\theta)(|k|):=\overline{P}_{el,j}(\theta)\otimes
        \id_f +P_{el,j}(\theta)\otimes \chi_{H_f+|k|\geq
        \rho_0}$.
        Moreover, we need the the operator $Q^{(\theta)}_0(z)$ on
        $\ran \Pel(\theta)$, defined by
        $[Q^{(\theta)}_0(z)\phi]\otimes\Omega:=Q^{(\theta)}(z)[\phi\otimes\Omega]$
        for all $\phi \in \ran \Pel(\theta)$.
        It is defined by the formula
        \begin{multline}\label{Eq:Q0Def}
            Q^{(\theta)}_0(z)\\=\sum_{\mu=1,2}\int\nolimits_{k\in \R^3} dk\chi_{|k|\geq \rho_0}  \Pel(\theta) \G{0}{1}{\theta}{k} \left
            [\frac{\Pel(\theta)}{E_j+e^{-i\vartheta}|k|-z}\right]
            \G{1}{0}{\theta}{k}\ \Pel(\theta)\\+
            \sum_{\mu=1,2}\int\nolimits_{k\in \R^3} dk \Pel(\theta) \G{0}{1}{\theta}{k} \left
            [\frac{\Pelbar(\theta)}{\Hel(\theta)+e^{-i\vartheta}|k|-z}\right]
            \G{1}{0}{\theta}{k}\ \Pel(\theta).
        \end{multline}
        We remark that both operators are analytic for $z\in
        \cA$. This follows from the fact that the resolvents in
        their definitions can be bounded uniformly in $z\in \cA$.
        (See the proof of Lemma \ref{Lm:QBound} for a proof in the
        case of $Q^{(\theta)}_0(z)$. The proof for
        $Q^{(\theta)}(z)$ is similar and uses additionally the
        spectral theorem for $H_f$.)

        Note that by assumption there
        exists a constant $c>0$ such that $\Im Z\geq c.$
        Since $Z$ is bounded, there are constants $a,b>0$ such that
         $\NumRan Z$ is localized as
        $\NumRan Z \subset A(c,a,b),$
        where $A(c,a,b):=ic+[-a,a]+i[0,b]$ (see Figure \ref{Fig:NumRanZ}). We set
        $\nu:=\min\{\vartheta, \arctan(c/(2a))\}$.

        Finally for $w \in \C$ and $r>0$ we define $D(w,r):=\{z\in \C|
        |z-w|<r\}$, and for $A\subset\C$ we set $D(A,r):=\{z\in \C|
        \dist(z,A)<r\}$.
        The notation $[z,w]$ denotes either the line segment between
        $z\in \C$ and $w\in \C$ or a linear contour from $z\in \C$ to $w\in
        \C$. Accordingly, $[z_1,w_1]+[z_2,w_2]$ is to be understood either as the
         sum of the sets $[z_1,w_1]\subset \C$ and $[z_2,w_2]\subset
        \C$ or as a generalized contour.

        \begin{proof}[Proof of Theorem \ref{Thm:TimeDecay}]
            First, we show that we can introduce a spectral cutoff
            with an error of $\cO(g)$:
            We choose a function
            $F\in C_0^\infty((E_j-\delta/2, E_j+\delta/2))$ with $0\leq
            F(x)\leq 1$ for all $x\in [E_j-\delta/2, E_j+\delta/2]$
            and
            $F(x)=1$ for all $x\in [E_j-\delta/4, E_j+\delta/4]$.
            By the almost analytic functional calculus
           \cite{HelfferSjostrand1989,Davies1995}
            \begin{align*}
                F(H_g)&=\frac{1}{\pi}\int dxdy \frac{\tilde F(z)}{\bar \partial
                z}(H_g-z)^{-1}\\
                &= \frac{1}{\pi}\int dxdy \frac{\tilde F(z)}{\bar \partial
                z}(H_0-z)^{-1}-\frac{1}{\pi}\int dxdy \frac{\tilde F(z)}{\bar \partial
                z}(H_g-z)^{-1}W_g (H_0-z)^{-1}
            \end{align*}
            where $\tilde F\in C_0^\infty(\C)$ is an almost analytic extension of
            $F(x)$ with $|\tfrac{\tilde F(z)}{\partial  \bar z}|= \cO(|\Im z|^2)$.
            Since
            \begin{equation*}
                \frac{1}{\pi}\int dxdy \frac{\tilde F(z)}{ \partial\bar
                z}\langle e^{isH_g}\Phi_1
                ,(H_0-z)^{-1}\Phi_2\rangle=\langle e^{isH_g}\Phi_1
                ,\Phi_2\rangle
            \end{equation*}
            and
            \begin{multline*}
                |\frac{1}{\pi}\int dxdy \langle e^{isH_g}\Phi_1
                ,\frac{\tilde F(z)}{\partial\bar
                z}(H_g-z)^{-1}W_g(H_0-z)^{-1} \Phi_2\rangle|\\\leq
                \frac{1}{\pi}\int dxdy \|\Phi_1\|
                |\frac{\tilde F(z)}{\bar \partial
                z}|\|(H_g-z)^{-1}\| |(E_j-z)^{-1}| \|W_g\Phi_2\|\leq
                C g,
            \end{multline*}
            we find that
            $\langle \Phi_1,
            e^{-isH_g} \Phi_2\rangle=\langle \Phi_1,
            e^{-isH_g}F(H_g)\Phi_2\rangle+\cO(g).$

            Analogous to \cite{Hunziker1990} and \cite{Bachetal1999S}, we can write
            \begin{multline*}
                \langle \Phi_1, e^{-isH_g}F(H_g)\Phi_2\rangle
                =-\frac{1}{2\pi
                i}\lim_{\epsilon\downarrow 0}\int
                d\lambda e^{-i\lambda s}F(\lambda)
                [f(0, \lambda-i\epsilon)-  f(0, \lambda+i\epsilon)]
                \\=-\frac{1}{2\pi
                i}\int
                d\lambda e^{-i\lambda s}F(\lambda)
                [f(\overline{\theta}, \lambda)-  f(\theta, \lambda)],
            \end{multline*}
            where $f(\theta, \lambda):=\langle \psi_1(\overline{\theta}),
            \frac{1}{H_g(\theta)-\lambda}\psi_2(\theta)\rangle$
            with $\psi_i(\theta):=\phi_i(\theta)\otimes \Omega$ and
            $\phi_i(\theta):=\cU_{el}(\theta)\phi_i$. We used Stone's
            theorem in the first step. In the second step we used the analyticity of $H_g(\theta)$ and
            the fact that $H_g(\theta)$ has no spectrum in the interval $[E_j-\delta/2,E_j+\delta/2]$
             (see \cite[Theorem 3.2]{Bachetal1999S} and also Corollary \ref{Cor:OpInv} below).

            Noting that $\langle \psi_1(\overline{\theta}),
            \frac{1}{H_g(\theta)-\lambda}\psi_2(\theta)\rangle=\langle\psi_1(\overline{\theta}),
            \cF[\lambda]^{-1}\psi_2(\theta)\rangle$ (see \cite[Formula (IV).14]{Bachetal1998Q} and also Lemma \ref{Lm:FeshExist}) and using the resolvent equation, we
            obtain
            \begin{multline*}
                f(\theta, \lambda)
                =\langle\psi_1(\overline{\theta}), \cF[\lambda]^{-1}\psi_2(\theta)\rangle\\
                =\langle\phi_1(\overline{\theta}), [E_j-\lambda-g^2Q_0^{(\theta)}(\lambda)]^{-1} \phi_2(\theta)\rangle
                -\langle\psi_1(\overline{\theta}),[E_j-\lambda-g^2Q_0^{(\theta)}(\lambda)]^{-1}\otimes \idf\\\times
                [\cF[\lambda]
                -(E_j-\lambda+e^{-\theta}\idel\otimes H_f-g^2
                Q^{(\theta)}(\lambda))\PO(\theta)]\\\times
                \cF[\lambda]^{-1}\psi_2(\theta)\rangle
                =:\tilde{f}(\theta, \lambda)+B(\theta,\lambda),
            \end{multline*}
            where $\tilde{f}(\theta, \lambda)$ is the first term in the
            sum.
            The strategy is now to move the contour for the first
            term in order to pick up a pole contribution (see Figure \ref{Fig:Contour}), and to
            estimate the second term on the real axis:
            \begin{eqnarray*}
            \lefteqn{\int
                d\lambda e^{-i\lambda s}F(\lambda)
                [f(\overline{\theta}, \lambda)-  f(\theta,
                \lambda)]}
                \\
                &=&\int
                d\lambda e^{-i\lambda s}F(\lambda)
                [B(\overline{\theta}, \lambda)-  B(\theta,
                \lambda)]
                +\int_{C_1+C_5} dz e^{-i z s}F(z)
                [\tilde f(\overline{\theta}, z)-  \tilde f(\theta,
                z)]\\
                &+&\int_{C_2+C_3+C_4} dz e^{-i z s}
                [\tilde f(\overline{\theta}, z)-  \tilde f(\theta,
                z)] - \int_{C_0} dz e^{-i z s}
                [\tilde f(\overline{\theta}, z)-  \tilde f(\theta,
                z)],
            \end{eqnarray*}
            where  we set $C:=C_1+C_2+C_3+C_4+C_5$, with
            $C_1:=[E_j-\delta/2, E_j-\delta/4]$, $C_2:=[E_j-\delta/4,
            E_j-\delta/4-ig^{2-\epsilon}/2]$,
            $C_3:=[E_j-\delta/4-ig^{2-\epsilon}/2,
            E_j+\delta/4-ig^{2-\epsilon}/2]$, $C_4:=[
            E_j+\delta/4-ig^{2-\epsilon}/2, E_j+\delta/4]$ and $C_5:=[E_j+\delta/4,
            E_j+\delta/2]$. $C_0$ is a suitable contour to pick up the pole contribution from $\tilde f(\theta,
                z)$. The analyticity properties required for this process will be discussed below.
                 Note that the contour $C$ cannot simply
                be moved down much further, since $Q_0^{(\theta)}(z)$
                may have singularities outside of $\cA$.
             \begin{figure}[h]
            \centering
            \includegraphics[width=\textwidth,clip=true]{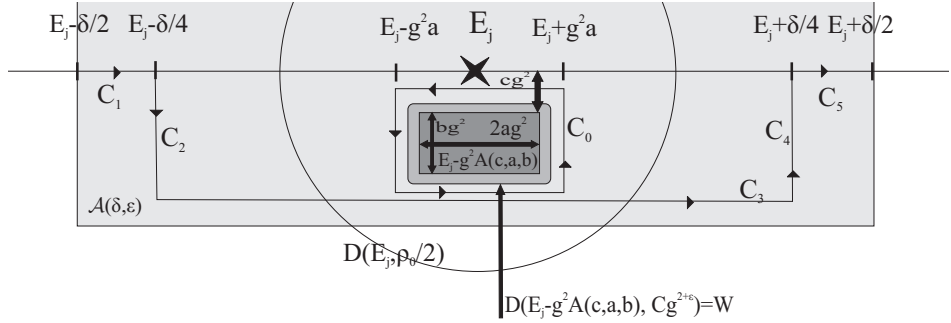}
            \caption{The integration contour}
            \label{Fig:Contour}
            \end{figure}

            \emph{Estimates on the real axis:} We divide the
            integration interval $[E_j-\delta/2, E_j+\delta/2]$ into two
            parts: On $[E_j-\delta/2, E_j+\delta/2]\setminus (E_j-\rho_0/2,
            E_j+\rho_0/2)$ we use Lemma \ref{Lm:QAppr} and Lemma \ref{Lm:NumRanContr} to
            obtain
            $
                |B(\theta,\lambda)|\leq C\cdot
                \frac{g^{2+\epsilon}}{(\sin\vartheta)^2(|\lambda-E_j|-Cg^2)^2}.
            $
            Since
            \begin{multline*}
                (\sin\vartheta)^{-2}\int_{\frac{g^{2-2\epsilon}}{2}}^\infty
                \frac{d\lambda}{(\lambda-Cg^2)^2}=
                (\sin\vartheta)^{-2} g^{-2}\int
                _{\frac{g^{-2\epsilon}}{2}}^\infty\frac{d\lambda}{(\lambda-C)^2}\\=
                (\sin\vartheta)^{-2} g^{-2}\frac{1}{g^{-2\epsilon}/2-C}=\cO(\vartheta^{-2} g^{-2+2\epsilon}),
            \end{multline*}
            we see that the error term for this region is of the
            order $\vartheta^{-2}g^{3\epsilon}$.
            On $(E_j-\rho_0/2, E_j+\rho_0/2)$ we estimate using
            Lemma \ref{Lm:QAppr} and
            Lemma \ref{Lm:FeshEst}:
            $|B(\theta,\lambda)|\leq C\vartheta^{-2}\cdot
            \frac{g^{2+\epsilon}}{(E_j-\lambda)^2+c^2g^4}$. Since
            $\int d\lambda
            \frac{g^{2+\epsilon}}{(E_j-\lambda)^2+c^2g^4}$ is easily
            seen to be of order $g^\epsilon$, and the same estimates hold for $B(\bar\theta,\lambda)$, the estimate on
            the real axis is proven.

            \emph{Estimates on the contour $C$:} We  estimate
            the integral $\int_C |e^{-isz}| |\tilde f(\overline{\theta}, z)- \tilde f(\theta,
            z)||dz|.$
            Note that
            \begin{multline*}\tilde f(\theta,
            z)\\=\frac{1}{E_j-z}\langle \phi_1(\overline\theta),\phi_2(\theta)\rangle+g^2\langle \phi_1(\overline\theta),
            \frac{1}{E_j-z} Q_0^{(\theta)}(z) \frac{1}{E_j-z-g^2
            Q_0^{(\theta)}(z)} \phi_2(\theta)\rangle.\end{multline*} Thus, the
            zero
            order terms of $f(\theta,
            z)$ and $f(\overline{\theta},
            z)$ cancel each other, and it suffices to show that the
            higher order terms are at least of order $g^{\epsilon}$.
            Since  $Q_0^{(\theta)}(z)$ is uniformly
            bounded in $z\in\cA$ by Lemma \ref{Lm:QBound}, we  estimate using Corollary
            \ref{Cor:NumRanQ0} (see Figure \ref{Fig:NumRanLoc})
            \begin{multline*}
                g^2|\langle\phi_1(\overline{\theta}),
                \frac{1}{E_j-(\lambda-ig^{2-\epsilon})} Q_0^{(\theta)}(\lambda-ig^{2-\epsilon})\\\times \frac{1}{E_j-(\lambda-ig^{2-\epsilon})-g^2
                Q_0^{(\theta)}(\lambda-ig^{2-\epsilon})} \phi_2(\theta)\rangle| \leq C\cdot
                \frac{g^2}{(E_j-\lambda)^2+(g^{2-\epsilon})^2}.
            \end{multline*}
            Thus the
            integral along $C_3$ of the above expression is easily
            seen to be of order $g^\epsilon$. The integral over the
            remaining contour is of order $g^2$, since
            $\dist(z, E_j)$ can be estimated independently of $g$ along this part of
            the
            contour. The integral of $\tilde f(\overline{\theta},
            z)$ can be estimated in the same way.

            \emph{Estimates on the pole term:} Since
            $Q_0^{(\theta)}(z)$ is uniformly bounded for $z\in \cA$ by Lemma \ref{Lm:QBound},
            the function $ \tilde f(\theta, z)$ has no poles in $\cA\setminus D(E_j,
            \rho_0/2)$. It follows by Lemma \ref{Lm:FeshEst} that
            $E_j-z-g^2Q_0^{(\theta)}(z)$ is bounded invertible if $z\in  D(E_j,
            \rho_0/2)\setminus (E_j-g^2A(c,a,b)+D(0, C_1\cdot
            g^{2+\epsilon}))\subset D(E_j,
            \rho_0/2)\setminus [
           \NumRan
            (E_j-g^2Z(0))+D(0, C_1\cdot g^{2+\epsilon})]
            $, i.e., all poles of $ \tilde f(\theta, z)$ are in the
            set $W:=E_j-g^2A(c,a,b)+D(0, C_1\cdot
            g^{2+\epsilon})$.

            Moreover, by Lemma \ref{Lm:FeshEst} we have
             the estimate
            $\|(E_j-z-g^2Q_0^{(\theta)}(z))^{-1}\| \leq\linebreak C \dist(z, \NumRan(E_j-g^2Z))^{-1}$
             for some $C>0$ if $z\in D(E_j, \rho_0/2)\setminus W$.
             In order
            to estimate the pole terms, we choose a contour $C_0$ around
            $W$ such that the length of the contour and its distance to $W$ are of
            order $g^2$. A possible choice is
            $C_0=[E_j+g^2(-(a+c/2)-ic/2),E_j+g^2(+(a+c/2)-ic/2)]+[E_j+g^2((a+c/2)-ic/2),E_j+g^2(+(a+c/2)-i(b+3c/2))]+
            [E_j+g^2(+(a+c/2)-i(b+3c/2)),E_j+g^2(-(a+c/2)-i(b+3c/2))]+[E_j+g^2(-(a+c/2)-i(b+3c/2)),E_j+g^2(-(a+c/2)-i(c/2))]$.
            We now use the expansion
            \begin{multline*}   \langle \phi_1(\overline{\theta}),
            (E_j-z-g^2Q_0^{(\theta)}(z))^{-1}\phi_2(\theta)\rangle=
            \langle \phi_1(\overline{\theta}),
            (E_j-z-g^2Z(\theta))^{-1}\phi_2(\theta)\rangle\\+g^2\langle \phi_1(\overline{\theta}),
             (E_j-z-g^2Q_0^{(\theta)}(z))^{-1}(Q_0^{(\theta)}(z)-Z(\theta))(E_j-z-g^2Z(\theta))^{-1}\phi_2(\theta)\rangle.\end{multline*}
             The integral over the first term gives the claimed
             leading term, the second term is of order $g^\epsilon$
             by Corollary \ref{Cor:ZAppr} and Lemma \ref{Lm:FeshEst}.

             Since by the above considerations (with $\theta$ replaced by $\bar
             \theta$)
              the function $\tilde f(\overline{\theta}, z)$ has
             no poles in the lower half-plane, there is no pole
             contribution from this function (see also Remark \ref{Rem:ZConj}).
        \end{proof}

    \section{Technical Lemmas}\label{sec:3}

        As in \cite{Bachetal1999S}, we need estimates on the
        numerical range and on the norm of the inverse of
        various operators. We make use of numerous results shown by
        Bach, Fr\"{o}hlich, and Sigal \cite{Bachetal1999S}, which are
        summarized in Appendix A.
        We use the following definitions from \cite{Bachetal1999S}:
        For $\eta>0$
        such that $E_j+\delta/2< \Sigma-\eta$ we define
        $P_{disc}(\theta):=\sum_{i: E_i\leq \Sigma-\eta}
        P_i(\theta)$ and
        $\overline{P}_{disc}(\theta):=1-P_{disc}(\theta)$.

        Since the operator valued function $Q^{(\theta)}_0$ is
        relevant for the location of the pole term in the time decay
        estimates, we need certain properties:

        \begin{lemma}\label{Lm:QBound} Let $\vartheta$ sufficiently
        small and
        $g^{\epsilon}/sin\vartheta\leq 1/2$. Then
            $Q_0^{(\theta)}(z)$ is uniformly bounded for $z\in
            \cA$.
        \end{lemma}
        \begin{proof}
            The proof follows \cite[Chapter IV]{Bachetal1998Q}, using, however, the
            following estimates:
            For the first summand in equation \eqref{Eq:Q0Def}, we use the estimate
            $|e^{-i\vartheta}|k|+E_j-z|\geq |\Im(
            e^{-i\vartheta}|k|+E_j-z)|\geq
            |\sin\vartheta|k| -g^{2-\epsilon}|\geq
            |\sin\vartheta|\cdot||k|-\rho_0/2|\geq
            1/2\sin\vartheta |k|$, for $|k|\geq \rho_0$.
            For the second summand in \eqref{Eq:Q0Def}, observe that for all $E_i$
            with $i\neq j$ we have
            $|E_i+e^{-\theta}|k|-z|\geq
            \sin\vartheta\delta/2-g^{2-\epsilon}\geq 1/4 \delta
            \sin\vartheta$ and that by Lemma \ref{Lm:BachBound2}
            $$\left\|(\Hel(\theta)-(z-e^{-i\vartheta}|k|))^{-1}\overline{P}_{disc}(\theta)\right\|\leq
            \frac{2}{\Sigma-\eta-\Re z+\cos\vartheta |k|}.$$
        \end{proof}
        This has the following immediate corollary:
        \begin{cor}\label{Cor:NumRanQ0}
            There exists a constant $C>0$ such that for all
            $z\in\cA$
            $$\NumRan(E_j-g^2Q_0^{(\theta)}(z))\subset D(E_j,
            C\cdot g^2).$$
        \end{cor}
        We use the following lemma to estimate  the inverse  of
        the
        Feshbach operator:
        \begin{lemma}\label{Lm:GenInvEst}
            Suppose $A$ is a bounded operator on a Banach space and let
            $A'$ be similar to $A$, i.e., there exists a
            bounded, bounded invertible operator $G$ such that
            $A'=GAG^{-1}$. Moreover, let $B$ be a another bounded
            operator. Then for any $q>1$ and for all $z\notin D(\NumRan(A), q\cdot \|B\| \|G  \| \cdot \| G^{-1}
            \|  )$ the following estimate holds:
            $$\left\|(A'+B-z)^{-1}\right\|\leq \|G  \| \cdot \|
            G^{-1}\| \frac{q}{q-1}
                \cdot \dist(z,
            \NumRan(A))^{-1}.$$
            In particular,
            $\sigma(A'+B)\subset D(\NumRan(A), q\cdot \|B\| \|G  \| \cdot \|
            G^{-1} \|  ).$

        \end{lemma}
        \begin{proof}
            First, observe that for all $z\notin \NumRan(A)$ by similarity
            $$\|(A'-z)^{-1}\|\leq \|G  \| \cdot \| G^{-1}
            \|\cdot \|(A-z)^{-1}\|\leq \|G  \| \cdot \| G^{-1}
            \|\cdot \dist(z,
            \NumRan(A))^{-1} .$$
            By a series expansion we obtain for
            $z\notin D(\NumRan(A), q\cdot \|B\|\cdot\|G  \| \cdot \| G^{-1}
            \|)$
            $$(A'+B-z)^{-1}=(A'-z)^{-1}\sum_{n=0}^\infty
            [-B(A'-z)^{-1}]^n.$$
            Taking the norm of both sides implies the claim.
        \end{proof}
        Following \cite{Bachetal1999S},  we
        control the Feshbach operator $\cF$ for $z\in D(E_j, \rho_0/2)$ as follows (see Figure
        \ref{Fig:NumRanZ}):
        \begin{lemma}\label{Lm:FeshEst} Let $0<\vartheta<\theta_0$
        and $0<g\ll \vartheta$ small enough. Then the following statements hold:
            \begin{itemize}
            \item[a)] There are constants $C_1, C_2>0$ such that $\cF[z]$ is bounded invertible for all $z\in
            D(E_j, \rho_0/2)\setminus D(\NumRan (E_j-g^2Z(0)\otimes \idf
            +e^{-\theta}\idel\otimes H_f)|_{\ran \PO(0)}, C_1\cdot g^{2+\epsilon })$, and for $\lambda \in [E_j-\rho_0/2,
            E_j+\rho_0/2]$ the estimate
            \begin{equation*}
                \\ \left\|\cF[\lambda]^{-1}\right\|\leq
                \frac{C_2}{\sin\nu\sqrt{(E_j-\lambda)^2+cg^4}}
            \end{equation*}
            holds. The same holds for $(E_j-z-g^2Q^{(\theta)}(z)\otimes \idf
            +e^{-\theta}\idel\otimes H_f)|_{\ran \PO(\theta)}$.

            \item[b)] There is a constant $C>0$ such that
            for all $z\in
            \C\setminus \NumRan (E_j-g^2Z(0))|_{\ran \Pel(0)}$ the operator
            $(E_j-z-g^2Z(\theta))|_{\ran \Pel(\theta)}$  is bounded
            invertible and fulfills the estimate
            \begin{multline}\|[(E_j-z-g^2Z(\theta))|_{\ran \Pel(\theta)}]^{-1}\|\\\leq \frac{C}{\dist(z,\NumRan (E_j-g^2Z(0))|_{\ran \Pel(0)})
            }.\end{multline}
            There are constants $C_1, C_2>0$ such that for all $z\in
            D(E_j, \rho_0/2)\setminus \linebreak D(\NumRan (E_j-g^2Z(0))|_{\ran \Pel(0)}, C_1\cdot g^{2+\epsilon
            })$ the operator $(E_j-z-g^2Q_0^{(\theta)}(z))|_{\ran
            \Pel(\theta)}$ is bounded invertible and fulfills
            \begin{multline}\|[(E_j-z-g^2Q_0^{(\theta)}(z))|_{\ran \Pel(\theta)}]^{-1}\|\\\leq \frac{C_2}{\dist(z,\NumRan (E_j-g^2Z(0))|_{\ran \Pel(0)})
            }.\end{multline}
            \end{itemize}

        \end{lemma}
        \begin{figure}[ht]
        \centering
        \includegraphics[width=\textwidth,clip=true]{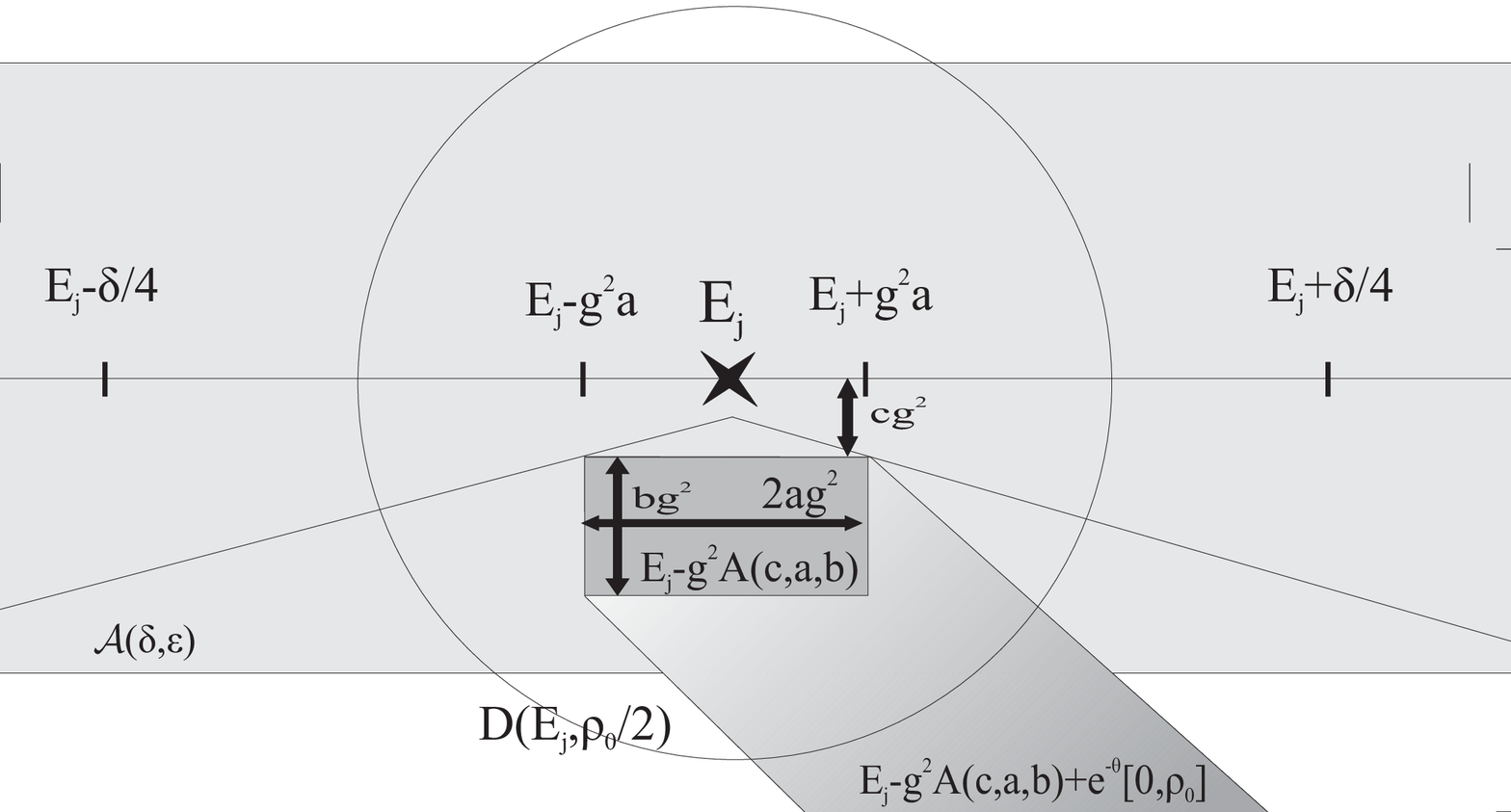}
        \caption{The numerical ranges of the operators $E_j-g^2 Z(0)$ and $\NumRan(E_j-g^2Z(0)\otimes \idf
            +e^{-\theta}\idel\otimes H_f)|_{\ran \PO(0)}$ }
        \label{Fig:NumRanZ}
        \end{figure}
        \begin{proof}
                By similarity (cf. Remark \ref{Rem:ZSim}) we obtain
                immediately for some $C_3>0$ that
            \begin{multline*}
                \|[(E_j-z-g^2Z(\theta)\otimes \idf
                +e^{-\theta}\idel\otimes H_f)|_{\ran \PO(\theta)}]^{-1}\|\\\leq C_1\cdot \dist(z, \NumRan(E_j-g^2Z(0)\otimes \idf
                +e^{-\theta}\idel\otimes H_f)|_{\ran \PO(0)})^{-1}.
            \end{multline*}
            By Lemma \ref{Lm:QAppr}, Corollary
            \ref{Cor:ZAppr}, and Lemma \ref{Lm:GenInvEst} there are constants $C_1, C_2>0$ such that
            \begin{multline}
            \|\cF[z]^{-1}\|\\\leq C_2 \dist(z, \NumRan(E_j-g^2Z(0)\otimes \idf
            +e^{-\theta}\idel\otimes H_f)|_{\ran \PO(0)})^{-1}
            \end{multline}
            follows for
                $z \notin D(\NumRan(E_j-g^2Z(0)\otimes \idf
            +e^{-\theta}\idel\otimes H_f)|_{\ran \PO(0)}), C_1\cdot g^{2+\epsilon})$.   It follows that
            $\NumRan[ (-g^2Z(0)\otimes \idf+ e^{-\theta}\idel\otimes H_f
            )|_{\ran \PO(0)}]\subset -g^2
            A(c,a,b)+e^{-\theta}[0, \rho_0]$ (see Figure \ref{Fig:NumRanZ}). By geometrical
            considerations, we see that this set is contained in
            the conical region
            $-i\frac{c}{2}g^2-i\{re^{i\phi}|-(\nu-\tfrac{\pi}{2})\leq\phi\leq \nu-\tfrac{\pi}{2},
            r\in[0,\infty)\}$. This, in turn, implies  the claim.
            The claims in b) follow by the same reasoning.
        \end{proof}
        We do not see that the estimate of Lemma \ref{Lm:FeshEst} a)
        is true for $\lambda\in [E_j-\delta/2,E_j+\delta/2]\setminus
        (E_j-\rho_0/2,E_j+\rho_0/2)$ as used in \cite[Proof of Theorem 3.5]{Bachetal1999S} (see also the remark after Lemma \ref{Lm:ZAppr}). Thus we bound
        $\cF[\lambda]^{-1}$ differently in that region in the next
        lemma.
        \begin{lemma}\label{Lm:NumRanContr}
        Let $0<\vartheta<\theta_0$
        and $0<g\ll \vartheta$ small enough.
            Then $\cF$ is bounded invertible for all $z\in
        \cA\setminus D(E_j, \rho_0/2)$ and there is a constant $C>0$ such that for $g$ small enough and $z\in
        \cA\setminus D(E_j, \rho_0/2)$  the
        numerical range of $\cF$ is localized as
        $$\NumRan(\cF+z)\subset D(E_j+e^{-\theta}[0,\rho_0], C g^2).$$
        In particular, for
        $\lambda\in [E_j-\delta/2, E_j+\delta/2]\setminus
        (E_j-\rho_0/2, E_j+\rho_0/2)$ the estimate
        $$\|\cF[\lambda]^{-1}\|\leq  \frac{1}{\sin\vartheta(|\lambda-E_j|-Cg^2)}  $$
        holds.
        Analogous statements hold with $\cF$ replaced by $E_j-z-g^2
        Q_0^{(\theta)}(z)$.
        \end{lemma}
        \begin{figure}[h]
        \centering
        \includegraphics[width=\textwidth,clip=true]{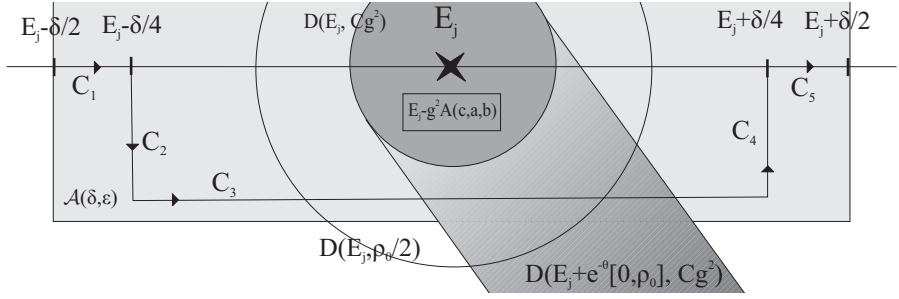}
        \caption{Global localization of the numerical range of $E_j+g^2
        Q_0^{(\theta)}(z)$ and of the Feshbach operator $\cF+z$.}
        \label{Fig:NumRanLoc}
        \end{figure}

        \begin{proof}[Proof of Lemma \ref{Lm:NumRanContr}]
            We have by Lemma \ref{Lm:QAppr} that $\|\PO(\theta) W_g(\theta) \PO(\theta)\|=\cO(g^{2+\epsilon})$. Therefore,  it
                suffices to show that
            $$\|\PO(\theta) W_g(\theta) \PObar(\theta) [\PObar(\theta)(H_g(\theta)-z)\PObar(\theta)]^{-1}\PObar(\theta)
            W_g(\theta)\PO(\theta)\|=\cO(g^2).$$
            Following  \cite[Proof of
            Lemma 3.14]{Bachetal1999S}, we use a Neumann expansion:
            \begin{multline}\label{EQ:Lm:NumRanContr:Pr:0}
                \PObar(\theta) [\PObar(\theta)(H_g(\theta)-z)\PObar(\theta)]^{-1}\PObar(\theta)
                \\=
                 \sum_{n=0}^\infty  \PObar(\theta)[\PObar(\theta)(H_0(\theta)-z)\PObar(\theta)]^{-1}\PObar(\theta)\\\times
                    \left[- \PObar(\theta)W_g(\theta)\PObar(\theta)
                \left[\PObar(\theta)(H_0(\theta)-z)\PObar(\theta)\right]^{-1}\PObar(\theta)
                \right]^n
            \end{multline}
            This expansion is valid if $z\in \cA$ with $\Im z\geq
            C$ for some $C> 0$ (independent of $g$.)
            We define
            $B_\theta(\rho):=\Hel(\theta)\otimes \idf-E_j+e^{-\theta}(\idel\otimes H_f+\rho)$ as in \cite{Bachetal1999S}.
            The right handside of equation
            \eqref{EQ:Lm:NumRanContr:Pr:0} is equal to
            \begin{multline}\label{EQ:Lm:NumRanContr:Pr:2}
               \sum_{n=0}^\infty |B_\theta(\rho_0)|^{-1/2}
                |B_\theta(\rho_0)|^{1/2} \PObar(\theta)\left[\PObar(\theta)(H_0(\theta)-z)\PObar(\theta)\right]^{-1}\PObar(\theta)
                |B_{\bar{\theta}}(\rho_0)|^{1/2}\\\times
                    \Big[-|B_{\bar{\theta}}(\rho_0)|^{-1/2} W_g(\theta)
                    |B_\theta(\rho_0)|^{-1/2}
                |B_\theta(\rho_0)|^{1/2}\PObar(\theta)\\\times
                \left[\PObar(\theta)(H_0(\theta)-z)\PObar(\theta)\right]^{-1}\PObar(\theta)
                |B_{\bar{\theta}}(\rho_0)|^{1/2}\Big]^n
                    |B_{\bar{\theta}}(\rho_0)|^{-1/2}=: R(z).
            \end{multline}
        for all $z\in \cA$ with $\Im z\geq
        C$.
            By Lemma \ref{Lm:BachBound1} and Corollary
            \ref{Cor:BachBound1} the series in equation \eqref{EQ:Lm:NumRanContr:Pr:2} converges uniformly
            for $z\in \cA$ and is thus a holomorphic function of
            $z\in \cA$. Thus
            $$ \PObar(\theta)
            [\PObar(\theta)(H_g(\theta)-z)\PObar(\theta)]^{-1}\PObar(\theta)=R(z)$$
            for all $z\in \cA$ by holomorphic continuation and for all $z\in \cA$
            \begin{multline}\label{EQ:Lm:NumRanContr:Pr:1}
                \PO(\theta) W_g(\theta) \PObar(\theta) \left[\PObar(\theta)(H_g(\theta)-z)\PObar(\theta)\right]^{-1}\PObar(\theta)
                W_g(\theta)\PO(\theta)\\=
                 \PO(\theta)|B_{\bar\theta}(\rho_0)|^{1/2}|B_{\bar\theta}(\rho_0)|^{-1/2} W_g(\theta)R(z)W_g(\theta) |B_\theta(\rho_0)|^{-1/2}
                    |B_\theta(\rho_0)|^{1/2}\PO(\theta).
             \end{multline}
            Note that
            $\||B_{\theta}(\rho_0)|\PO(\theta)\|=\|B_{\theta}(\rho_0)\PO(\theta)\|$
            and
            $\|\PO(\theta)|B_{\bar\theta}(\rho_0)|\|=\|\PO(\theta)B_{\theta}(\rho_0)\|$.
            Thus, using $\|\PO(\theta)|B_{\bar\theta}(\rho_0)|^{1/2}\|\leq
            \|\PO(\theta)|B_{\bar\theta}(\rho_0)|\|\cdot
            \||B_{\bar\theta}(\rho_0)|^{-1/2}\|=\cO(\rho_0^{1/2})$,
            and counting the powers of
            $\rho_0$ in \eqref{EQ:Lm:NumRanContr:Pr:1}, the first claim
            follows. The estimate on the inverse follows by
            geometrical considerations. The claim on $E_j-z-g^2Q_0^{(\theta)}(z)$ follows from Lemma \ref{Lm:QBound}.
        \end{proof}
        Note that due to the appearance of the interaction
        $W_g(\theta)$ on both sides of the resolvent $[\PObar(\theta)(H_g(\theta)-z)\PObar(\theta)]^{-1}$ and due to the
        projections $\PO(\theta)$, the divergence of the resolvent
        for $\rho_0\rightarrow 0$ is completely eliminated (see also
        the remark after Lemma \ref{Lm:FeshExist}).

        We use the following corollary instead of \cite[Theorem
        3.2]{Bachetal1999S}.
        \begin{cor}\label{Cor:OpInv}
            Let $0<\vartheta<\theta_0$
            and $0<g\ll \vartheta$ small enough. Then
            $$\cA\setminus (E_j- D(g^2A(c,a,b), C\cdot
            g^{2+\epsilon})
            +e^{-\theta}[0,\rho_0])\subset \rho(H_g(\theta))$$
            for some $C>0$.
            In particular, the interval $[E_j-\delta/2, E_j+\delta/2]$
            is contained in the resolvent set $\rho(H_g(\theta))$.
        \end{cor}
        \begin{proof}
            By Lemma \ref{Lm:NumRanContr}, $\cF$ is bounded invertible for all
            $z\in
             \cA\setminus D(E_j, \rho_0/2)$.
             By Lemma \ref{Lm:FeshEst}, it is
             bounded invertible for all $z\in
             D(E_j, \rho_0/2)\setminus D(\NumRan (E_j-g^2Z(0)\otimes \idf
             +e^{-\theta}\idel\otimes H_f)|_{\ran \PO(0)}, C_1\cdot
            g^{2+\epsilon })$.
             Lemma
            \ref{Lm:FeshExist} implies the claim.
        \end{proof}
        \begin{appendix}

            \section{Estimates Taken from Bach, Fr\"{o}hlich, and Sigal \cite{Bachetal1999S}}
                 In the appendix, we quote some important technical
                lemmas from \cite{Bachetal1999S}, which we frequently
                use. We do not give their proofs, since they are very lengthy. However, for the
                orientation of the reader, we describe the
                essential points of the proofs in words.
                \subsection{Existence of the Feshbach Operator}

                First we need certain relative bounds on the
                interaction and bounds on the resolvent.
                \begin{lemmaA}[\cite{Bachetal1999S}, Lemma 3.8] \label{Lm:BachBound2}
                    Let $z\in \C$ with $\Re z< \Sigma-\eta$. Then, for
                    $|\theta|(1+(\Sigma-\eta-\Re z)^{-1})$
                    sufficiently small, $\Hel(\theta)-z$ is invertible
                    on $\ran \bar P_{disc}(\theta)$ and
                    $$\left\|(\bar P_{disc}(\theta)\Hel(\theta)\bar P_{disc}(\theta)-z)^{-1}\bar P_{disc}(\theta)\right\|\leq 2(\Sigma-\eta-\Re z)^{-1}$$
                \end{lemmaA}
                \noindent This lemma is proved by using that the estimate
                holds for $\theta=0$  with constant one (instead of
                two)
                and using that $\Hel(0)-\Hel(\theta)$ is relatively
                $\Hel(0)$ bounded.
                We remind the reader that as in \cite{Bachetal1999S} we define
                $B_\theta(\rho):=\Hel(\theta)\otimes \idf-E_j+e^{-\theta}(\idel\otimes
                H_f+\rho)$. Note that
                $\cA\subset \rho(\PObar(\theta)H_0(\theta))$.
                \begin{lemmaA}[\cite{Bachetal1999S}, Lemma 3.11]\label{Lm:BachBound3}
                     There exists a constant  $C>0$ such that for $0<\vartheta<\theta_0$, for all $g$ with $0\leq
                        g\rho_0^{-1/2}\leq 1/3$ and $0<\rho_0\leq
                        (\delta/3)\sin\vartheta$, and for all $z\in \cA$
                        \begin{equation}\label{Lm:BachBound3:Eq:1}
                            \|B_\theta(\rho_0) \frac{\PObar(\theta)}{H_0(\theta)-z} \|\leq \frac{C}{\vartheta}.
                        \end{equation}
                \end{lemmaA}
                \noindent The proof of Lemma \ref{Lm:BachBound3} is based on Lemma  \ref{Lm:BachBound2},
                the fact that $\Hel(\theta)$ restricted to $P_{disc}(\theta)$ is
                similar to a self-adjoint operator, and various
                other estimates on the resolvent of $\Hel(\theta)$ as well as the application of the spectral theorem for
                $H_f$.

                \noindent The following Corollary was used in
                \cite{Bachetal1998S}:
                \begin{corA}\label{Cor:BachBound1}
                    There exists a
                    constant $C>0$ such that for $0<\vartheta<\theta_0$, all $g$ with $0\leq
                    g\rho_0^{-1/2}\leq 1/3$ and $0<\rho_0\leq
                    (\delta/3)\sin\vartheta$,  and for all $z\in \cA$
                    $$\||B_\theta(\rho_0)|^{1/2} \frac{\PObar(\theta)}{H_0(\theta)-z} |B_{\bar{\theta}}(\rho_0)|^{1/2}\|\leq \frac{C}{\vartheta}.$$
                \end{corA}
                \begin{proof}
                    By taking adjoints in equation \eqref{Lm:BachBound3:Eq:1}, we find that
                    $\| \frac{\PObar(\theta)}{H_0(\theta)-z} |B_{\bar{\theta}}(\rho_0)|\|\leq \frac{C}{\vartheta}.$
                    The claim follows by complex interpolation.
                \end{proof}

                \begin{lemmaA}[\cite{Bachetal1999S}, Lemma
                3.13]\label{Lm:BachBound1}
                      There is a constant $C>0$ such that for  $0<\vartheta<\theta_0$ sufficiently
                        small, $\theta_1,\theta_2\in \{\pm
                        i\vartheta\}$ and for all $\rho>0$
                        \begin{equation}\label{Lm:BachBound1:Eq:1}
                        \||B_{\theta_1}(\rho)|^{-1/2}W_g(\theta)|B_{\theta_2}(\rho)|^{-1/2}\|\leq
                        g
                        \frac{C}{\vartheta}(1+\rho^{-1/2}).
                        \end{equation}
                \end{lemmaA}
                \noindent The proof of Lemma \ref{Lm:BachBound1}  uses that $ \|A_{\kappa}(x)_-\psi\|\leq C \|H_f^{1/2}\psi\|$ and
                $ \|A_{\kappa}(x)_+\psi\|\leq C
                \|(H_f+1)^{1/2}\psi\|$ for some $0<C$ and all $\psi$
                in the domain of $H_f^{1/2}$, that $\Hel(0)-\Hel(\theta)$ is relatively
                $\Hel(0)$ bounded, and some other estimates. The
                term proportional to $\rho^{-1/2}$ is  due to the $+1$ in
                the bound for the creation operator
                $A_{\kappa}(x)_+$ and to the appearance of a similar constant in the estimate for
                the relative boundedness of $\Hel(0)-\Hel(\theta)$.
                Note that the symmetric form of the estimate
                \eqref{Lm:BachBound1:Eq:1} is essential. Estimates on
                $W_g(\theta)|B_{\theta_2}(\rho)|^{-1}$lead to  worse behavior as $\rho\rightarrow 0$.

                \begin{lemmaA}[\cite{Bachetal1999S}, Lemma
                3.14]\label{Lm:BachBound4}
                    There is a $C>0$ such that
                    for $\vartheta\in(0,\theta_0)$,
                    $\rho_0<(\delta/3)\sin\vartheta$, $0<g\rho_0^{-1/2}\ll\vartheta^2$,  and for all $z\in \cA$
                    the operator $\PObar(\theta) H_g
                    \PObar(\theta) -z$ is invertible on $\ran \PObar(\theta)$ and fulfills
                    $$\|[\PObar(\theta) H_g(\theta)
                    \PObar(\theta) -z]^{-1}\PObar(\theta)\|\leq \frac{C}{\vartheta \rho_0}.$$
                \end{lemmaA}
                \noindent The proof of Lemma \ref{Lm:BachBound4}  uses Corollary \ref{Cor:BachBound1}, Lemma
                \ref{Lm:BachBound1} and a Neumann series expansion.

                \begin{lemmaA}[\cite{Bachetal1999S}, Lemma
                3.15]\label{Lm:FeshExist}
                    Assume that
                    $\vartheta\in(0,\theta_0)$. Let $\rho_0<(\delta/3)\sin\vartheta$ and  $0<g\rho_0^{-1/2}\ll\vartheta^2$.
                      Then for all $z\in
                    \cA$ the Feshbach operator $\cF$ defined in equation
                    \eqref{Eq:Fesh} exists. If $z\in \cA$, then  $H_g(\theta)-z$ is bounded invertible if
                     and only if the Feshbach operator $\cF$ is bounded
                     invertible, and
                    the equation
                    \begin{multline}\label{Lm:FeshExist:Eq:1}
                        (H_g(\theta)-z)^{-1}=[\PO(\theta)-\PObar(\theta)(\PObar(\theta) H_g(\theta)
                    \PObar(\theta) -z)^{-1}\PObar(\theta)W_g(\theta) \PO(\theta)]\\
                    \times \cF^{-1} [\PO(\theta) -\PO(\theta)W_g(\theta) \PObar(\theta)(\PObar(\theta) H_g(\theta)
                    \PObar(\theta) -z)^{-1}\PObar(\theta)]\\+\PObar(\theta)[\PObar(\theta) H_g(\theta)
                    \PObar(\theta) -z]^{-1}\PObar(\theta)
                    \end{multline}
                    holds, where the left side exists if and only if
                    the right side exists.
                    Moreover, there is a constant $C>0$, independent of
                    $g$ and $\theta$, such that for all $z\in \cA$
                    \begin{equation}\label{Lm:FeshExist:Eq:2}
                    \|(\PObar(\theta) H_g(\theta)
                    \PObar(\theta) -z)^{-1}\PObar(\theta)W_g(\theta) \PO(\theta)\|\leq \frac{C g}{\vartheta
                    \rho_0^{1/2}}\end{equation}
                    and
                    \begin{equation}\label{Lm:FeshExist:Eq:3}
                    \|\PO(\theta)W_g(\theta) \PObar(\theta)(\PObar(\theta) H_g(\theta)
                    \PObar(\theta) -z)^{-1}\|\leq \frac{C g}{\vartheta
                    \rho_0^{1/2}}.\end{equation}
                \end{lemmaA}
                \noindent Equations \eqref{Lm:FeshExist:Eq:2} and
                \eqref{Lm:FeshExist:Eq:3} are proved similarly as
                Lemma \ref{Lm:BachBound4}. Together with Lemma \ref{Lm:BachBound4} they imply the existence
                of the Feshbach operator and the validity of
                Equation \eqref{Lm:FeshExist:Eq:1} (see
                \cite[Theorem IV.1]{Bachetal1998Q}). Note that the
                operator $W_g(\theta)$ in Formulas
                \eqref{Lm:FeshExist:Eq:2} and
                \eqref{Lm:FeshExist:Eq:3} reduces the divergence as
                $\rho_0\rightarrow 0$ in comparison to Lemma
                \ref{Lm:BachBound4}.

                \subsection{Approximations of the Feshbach Operator}

                The following lemma gives an approximation of
                the Feshbach operator globally for all $z\in \cA$ (see \cite{Bachetal1999S},
                Lemma 3.16, estimates on $Rem_0$ through $Rem_3$):
                \begin{lemmaA}\label{Lm:QAppr}
                    Let $0<\epsilon<1/3$ and $0<\vartheta<\theta_0$.
                    Then there is a constant $C\geq 0$ such that for all
                    $g>0$ sufficiently small  with
                    $\rho_0<(\delta/3)\ sin\vartheta$, and for all
                    $z\in \cA$
                    $$\|[\cF-(E_j-z+e^{-\theta}\idel\otimes H_f-g^2 Q^{(\theta)}(z))]\PO(\theta) \|\leq C g^{2+\epsilon}.$$
                    Moreover
                    $\|\PO(\theta)W_g(\theta)\PO(\theta)\|=\cO(g^{2+\epsilon}).$
                \end{lemmaA}
                The lengthy and technical proof of Lemma \ref{Lm:QAppr} is based on a Neumann series expansion, estimates similar to Lemma
                \ref{Lm:BachBound1}, and the pull-through formula.

                For $z$ sufficiently close to $E_j$,
                $Q^{(\theta)}(z)$ can be approximated by $\tilde Z(\alpha,\theta)$
                (see \cite[Lemma 3.16, Estimates on $Rem_4$ and
                $Rem_5$]{Bachetal1999S}).
                \begin{lemmaA} \label{Lm:ZAppr}
                Let $0<\epsilon<1/3$ and $0<\vartheta<\theta_0$.
                    Then there is a constant $C\geq 0$ such that for all
                    $g>0$ sufficiently small  with
                    $\rho_0<(\delta/3)\ sin\vartheta$, and for all
                    $z\in D(E_j, \rho_0/2)$
                    $$g^2\|Q^{(\theta)}(z)- \tilde Z(\alpha, \theta)\|\leq Cg^{2+\epsilon}.$$
                \end{lemmaA}
                The proof requires some additional estimates to
                eliminate the z-dependence of $Q^{(\theta)}(z)$.
                 However, we not see that Lemma \ref{Lm:ZAppr} holds for all $z\in
                \cA$, which seems to be used in \cite{Bachetal1999S}.

                Lemma \ref{Lm:ZAppr} and Equation
                \eqref{Eq:ZAlphaZero}
                imply
                \begin{corA}\label{Cor:ZAppr}
                    Under the assumptions of Lemma \ref{Lm:ZAppr}
                    $$g^2\|Q^{(\theta)}(z)- Z(\theta)\|\leq Cg^{2+\epsilon}.$$
                \end{corA}

                \begin{rem}\label{Rem:ZConj}
                    Note that in order to approximate the
                    Feshbach operator $\mathcal{F}_{\PO(\bar \theta)}(H_g(\bar
                    \theta)-z)$ for $\theta=i\vartheta$ with $\vartheta>0$, the $-i\epsilon$ in definition
                    \eqref{Eq:ZDef} has to be replaced by $+i\epsilon$. In
                    particular, when considering the spectral analysis
                    of this operator, the localization of the numerical
                    range and of the spectrum have to be reflected about
                    the real axis.
                \end{rem}
            \section{The Hydrogen Atom}
                In this section we discuss the applicability of the
                presented method to the hydrogen atom. In
                particular, we show that $\Im Z$ is strictly
                positive unless $j=1$. For compatibility with
                physics literature, we number the eigenvalues of the
                hydrogen atom according to the principal quantum
                number $n=i+1$. We denote the corresponding
                eigenvalues by $\mathfrak{E}_n$, i.e.,
                $\mathfrak{E}_n=E_i$ for all $i\geq 0$. We will
                ignore the (trivial) spin dependence of $Z=\tilde Z(0,0)$ in this appendix.

                \subsection{The Hydrogen Eigenfunctions}
                    We define the associated Laguerre polynomials
                    (see \cite[Formula (3.5)]{BetheSalpeter1957Q})
                    for $\lambda,\mu\in \N_0$ with $0\leq \mu\leq \lambda$
                    by
                    $$L^\mu_\lambda(r):=\left(\frac{d}{dr}\right)^\mu\left( e^r\left(\frac{d}{dr}\right)^\lambda\left( e^{-r} r^\lambda\right)
                    \right)$$
                    and set (see \cite[Formula (3.16)]{BetheSalpeter1957Q})
                    \begin{equation}\label{AppB:Eq:RadFunc}
                    R_{n,l}(r):=-\frac{1}{\sqrt{8}}\frac{(n-l-1)!^{1/2}}{(n+l)!^{3/2}(2n)^{1/2}}(2/n)^{3/2}e^{-r/(2n)}\left(
                    \frac{r}{n}\right)^l L_{n+l}^{2l+1}(r/n).
                    \end{equation}
                    Note that the Hamiltonian in
                    \cite{BetheSalpeter1957Q} has an additional
                    factor of $1/2$ in front of the Laplacian, so
                    that the radial functions and certain other
                    quantities have to be adapted accordingly.
                    We would like to warn the reader that there are
                    different conventions for the indices of the
                    associated Laguerre functions.

                    For $n\in \N$ and $l,m\in \Z$ with  $0\leq l\leq n-1$ and $-l\leq m\leq l$ the normalized eigenfunctions to the eigenvalue $\mathfrak{E}_n$ are
                    \begin{equation}\label{AppB:EFDef}u_{n,l,m}(r,\theta,\phi):=R_{n,l}(r)Y_{l,m}(\theta,\phi),\end{equation}
                    where
                    the $Y_{l,m}$ are spherical harmonics (see \cite[Section 1]{BetheSalpeter1957Q})  and
                     we introduced polar coordinates by
                    \begin{eqnarray*}
                        x&=&r\sin\theta \cos\phi\\
                        y&=&r\sin\theta \sin\phi\\
                        z&=&r\cos\theta
                    \end{eqnarray*}
                    with $0\leq \theta\leq \pi$ and $0\leq \phi\leq
                    2\pi$. Note that in this appendix $x$, $y$, and $z$
                    denote the cartesian coordinates of the
                    electron, contrary to the main part of the
                    paper, where $x$ and $z$ have different
                    meanings.
                    Moreover, note that
                    the eigenvalues $\mathfrak{E}_n$ are
                    $n^2$-fold degenerate.
                \subsection{Selection Rules for Dipole Transitions}
                    In this subsection, we give some important results from \cite{BetheSalpeter1957Q}. We define
                    \begin{equation}\label{AppB:Eq:RDef}
                    R_{n,l}^{n',l'}:=\int_0^\infty dr r^3
                    R_{n',l'}(r) R_{n,l}(r).
                    \end{equation}
                    These integrals have been evaluated by Gordon
                    \cite{Gordon1929} (see also \cite[Section
                    63]{BetheSalpeter1957Q}).
                    Below, we need  (see \cite[Formula (63.4)]{BetheSalpeter1957Q})
                    \begin{equation}\label{AppB:Eq:REval1}
                        |R_{2,1}^{n,0}|=2\cdot\sqrt{\frac{2^{15}n^9(n-2)^{2n-6}}{3(n+2)^{2n+6}}}.
                    \end{equation}
                    For the dipole moments $ (u_{n',l',m'}, z u_{n,l,m})$
                    one finds (see \cite[Formula
                    (60.11)]{BetheSalpeter1957Q})  for all $n,n'\in \N_0$ that
                    \begin{equation}
                        \label{AppB:Eq:Z3}(u_{n',l',m'}, z u_{n,l,m})=0\quad \text{unless $l'=l\pm1$ and $m'=m$}.
                    \end{equation}
                    Moreover, we will need the relation
                    \begin{equation}
                        \label{AppB:Eq:Z2}(u_{n', 0,
                        0}, z u_{2,1,0})=\sqrt{\frac{1}{3}}R_{2,1}^{n',
                        0}.
                    \end{equation}
                    The selection rules given in \cite[Formula
                    (60.11)]{BetheSalpeter1957Q} imply immediately
                    \begin{equation}
                        \label{AppB:Eq:XY5}(u_{n',l',m'}, x u_{n,l,m})=(u_{n',l',m'}, y u_{n,l,m})=0
                    \end{equation}
                    unless $l'=l\pm 1$ and $m'=m\pm 1$

                \subsection{The Imaginary Part of $Z$}
                    In this subsection, we show that the method
                    presented in this paper applies to the hydrogen
                    atom, except for the case $n=2$.
                    \begin{theoremA}
                        Fix $n\in \N$ and consider
                        \begin{multline*}
                            \Im Z=\frac{1}{6\pi}
                            \sum_{i=0}^{j-1}
                            (E_j-E_i)^3 \kappa(E_j-E_i)^2\\\times
                            \left[\Pel x  P_{el,i} x
                            \Pel+\Pel y  P_{el,i} y
                            \Pel+\Pel z  P_{el,i} z
                            \Pel\right]
                        \end{multline*}
                         for $j=n-1$ as in equation \eqref{Eq:ImZX}.
                         Then for all $l,m,l',m'\in \N_0$ with $0\leq l\leq
                        n-1$,  $-l\leq m\leq l$, $0\leq l'\leq
                        n-1$, and  $-l'\leq m'\leq l'$
                        $$(u_{n,l',m'}, \Im Z u_{n,l,m})=0$$
                        unless $l=l'$ and  $m=m'$,
                        and for all $l,m\in \N_0$ with $0\leq l\leq
                        n-1$,  $-l\leq m\leq l$
                        $$(u_{n,l,m}, \Im Z u_{n,l,m})>0$$
                        unless $n=2$. In particular, $\Im Z$ is
                        positive, unless $n=2$.
                    \end{theoremA}
                    \begin{proof}
                    \emph{Off-diagonal matrix elements:}
                    Since $\Im Z$ is invariant under
                    rotations, it is diagonal in the basis
                    $\{u_{n,l,m}|\,
                    0\leq l\leq n-1,\, -l\leq m\leq
                    l\}$. This can also be verified using
                    the explicit formulas for the dipole matrix elements in \cite[Section
            63]{BetheSalpeter1957Q} . Note that the
                    matrices $\Pel x  P_{el,i} x
                    \Pel$, $\Pel y  P_{el,i} y
                    \Pel$, and $\Pel z  P_{el,i} z
                    \Pel$ are not diagonal separately. We would like
                    to mention that also the real part is diagonal
                    in the basis $\{u_{n,l,m}|\,
                    0\leq l\leq n-1,\, -l\leq m\leq
                    l\}$.

                    \noindent\emph{Diagonal matrix elements}:
                    Let us first remark that the matrix
                    element $$(u_{2,0,0}, [P_{el,1} x  P_{el,0} x
                        P_{el,1}+P_{el,1} y  P_{el,0}
                        y
                        P_{el,1}+P_{el,1} z  P_{el,0}
                        z
                        P_{el,1}]u_{2,0,0})$$
                        vanishes by the selection rules
                        \eqref{AppB:Eq:XY5} and
                        \eqref{AppB:Eq:Z3}.

                        Suppose now that $n\geq 3$.
                        We have to prove that there is an
                        $i<j=n-1$ such that for all $\phi\in
                        \ran \Pel$
                        $$\sum_{\upsilon=x,y,z}\|P_{el,i}p_\upsilon\phi\|^2>0.$$
                        Since
                        $\Im Z$ is diagonal in the basis $\{u_{n,m,l}|l=0\ldots n-1,\, m=-l,\ldots, l\}$, it suffices to show
                        $$\sum_{\upsilon=x,y,z}\|P_{el,i}p_\upsilon u_{n,l,m}\|^2>0$$
                        for all $0\leq l\leq n-1$ and
                        $-l\leq m\leq l$. For the case
                        $l=0$, $m=0$ it follows from
                        equations \eqref{AppB:Eq:Z2} and \eqref{AppB:Eq:REval1}
                        that the transition
                        $(n,0,0)\rightarrow (2,1,0)$ is an
                        allowed electric dipole transition,
                        since $z^{n,0,0}_{2,1,0}>0$.
                        Consequently $(u_{n,0,0}, \Im Z
                        u_{n,0,0})>0$.

                        Thus, it suffices to consider the case
                        $l>0$. The proof is by contradiction. Assume  that
                        $\sum_{\upsilon=x,y,z}\|P_{el,i}p_\upsilon u_{n,l,m}\|^2=0$
                        for all $i<j=n-1$ and some $l,m$.
                        This would imply that for
                        $\upsilon=x,y,z$
                        $$(p_\upsilon
                        u_{n,l,m}, \Hel p_\upsilon
                        u_{n,l,m})\geq E_j(p_\upsilon
                        u_{n,l,m},  p_\upsilon
                        u_{n,l,m}).$$
                            For $l>0$, it is easy to see by Equation
                         \eqref{AppB:Eq:RadFunc}
                         that $p_\upsilon u_{n,l,m}\in \dom(\Hel)$
                         and, using partial integration and the fact that
                         $u_{n,l,m}(0)=0$,  we see that
                         $$\sum_{\upsilon=x,y,z}(p_\upsilon
                         u_{n,l,m}, \Hel p_\upsilon
                         u_{n,l,m})=E_j\sum_{\upsilon=x,y,z}(p_\upsilon
                         u_{n,l,m},  p_\upsilon
                         u_{n,l,m}).$$
                            Thus, we conclude by the variational principle that
                        $$\Hel p_\upsilon u_{n,l,m}=E_j
                        p_\upsilon
                        u_{n,l,m}.$$
                        However,
                        $$E_j p_\upsilon
                        u_{n,l,m}=\Hel p_\upsilon u_{n,l,m}=E_j
                        p_\upsilon
                        u_{n,l,m}+[\Hel,p_\upsilon]u_{n,l,m}$$ for $\upsilon=x,y,z$
                        and
                        $$[\Hel,p_x]=-i\frac{x}{r^3}$$
                        with $r=\sqrt{x^2+y^2+z^2}$,
                        so that we arrive at a
                        contradiction.

                    \end{proof}
            \subsection{Numerical Illustration}
                In this subsection we give explicit numerical values
                for the matrix $\Im Z$ for the case $n=3$ setting the cutoff function $\kappa$ identically equal to one.
                Using Maple and the explicit form of the
                eigenfunctions in equation \eqref{AppB:EFDef}, we
                calculate the matrices $P_{el,0}xP_{el,2}$ and
                $P_{el,1}xP_{el,2}$ as well as the corresponding
                matrices for the coordinates $y$ and $z$, where
                $P_{el,0}$ is the projection onto the groundstate,
                $P_{el,1}$ the projection onto the eigenspace
                belonging to $\mathfrak{E}_2$, and $P_{el,2}$ the projection onto the eigenspace
                belonging to $\mathfrak{E}_3$. With these matrices,
                we calculate $\Im Z$ according to equation
                \eqref{Eq:ImZX}. The
                numerical values for other principal quantum numbers
                could be calculated in the same way.

                The matrix  $\Im Z$ (and also $Z$) is diagonal in the basis $\{u_{3,l,m}|\,
                    0\leq l\leq 2,\, -l\leq m\leq
                    l\}$. The diagonal elements depend only on $l$,
                    but not on $m$. We find $(u_{3, 0,
                0}, \Im Z u_{3, 0, 0})=\frac {192}{1953125}$, $(u_{3, 1,m}, \Im Z u_{3, 1,m})=\frac {738423}{250000000}$ for $-1\leq m\leq 1$,
                  and $(u_{3, 2,m}, \Im Z u_{3, 2,m})=\frac {49152}{48828125}$ for $-2\leq m\leq 2$.
                Let us remark that the eigenvalues of $2\cdot(2\alpha^5 m\mathfrak{c}^2/\hbar)\Im Z$ are
                precisely the inverse lifetimes $\tau_{n,l,m}^{-1}$ of the corresponding
                eigenstates of the hydrogen atom. The additional
                factor two is due to the fact that lifetimes are
                defined via survival probabilities and not via
                survival amplitudes.
                Inserting $\alpha=7.29735\cdot 10^{-3}$, $m=9.10939\cdot 10^{-31}kg$, $\mathfrak{c}=2.99792\cdot 10^{8}m/s$ and
                $\hbar=1.05457\cdot 10^{-34}Js$ we find
                $\tau_{3,0,0}=1.58303\cdot 10^{-7}s$,
                $\tau_{3,1,m}=5.26860\cdot 10^{-9}s$ for
                $-1\leq m\leq 1$, and $\tau_{3,2,m}=1.54593\cdot 10^{-8}s$ for
                $-2\leq m\leq 2$. Experimental values for these lifetimes are not very precise.  We
                quote a value of $\tau_{3,1,m}= (5.5  \pm 0.2) \times 10^{-9}s$ given in
                \cite{Chuppetal1968}.  \cite{BickelGoodman1966} find
                a value of $\tau_{3,1,m}= (5.58  \pm 0.13) \times
                10^{-9}s$  and  \cite{Ethertonetal1970} find $\tau_{3,1,m}= (5.41  \pm 0.18) \times
                10^{-9}s$. Notice that the experimental values  are in
                reasonable agreement with the calculated value.

        \end{appendix}
        \begin{ack}
            M.H. wishes to thank I.H. for the hospitality of the
            Mathematics Department of the University of Virginia and
            V. Bach for interesting conversations.
            He acknowledges support by the Deutsche
            Forschungsgemeinschaft (DFG), grant no. SI
            348/12-2. His
            stay at the University of Virginia was supported by a
            ``Doktorandenstipendium''
            from the German Academic Exchange Service
            (DAAD), which he gratefully acknowledges.
        \end{ack}



\begin{thebibliography}{10}

\bibitem{Bachetal1995}
V.~Bach, J.~Fr\"{o}hlich, and I.~M. Sigal.
\newblock Mathematical theory of nonrelativistic matter and radiation.
\newblock {\em Lett. Math. Phys.}, 34(3):183--201, 1995.

\bibitem{Bachetal1998S}
V.~Bach, J.~Fr\"{o}hlich, and I.M. Sigal.
\newblock Spectral analysis for systems of atoms and molecules coupled to the
  quantized radiation field.
\newblock Preprint, mp\_arc 98-728, 1998.

\bibitem{Bachetal1998Q}
Volker Bach, J\"{u}rg Fr\"{o}hlich, and Israel~Michael Sigal.
\newblock Quantum electrodynamics of confined nonrelativistic particles.
\newblock {\em Adv. Math.}, 137(2):299--395, 1998.

\bibitem{Bachetal1999S}
Volker Bach, J\"{u}rg Fr\"{o}hlich, and Israel~Michael Sigal.
\newblock Spectral analysis for systems of atoms and molecules coupled to the
  quantized radiation field.
\newblock {\em Comm. Math. Phys.}, 207(2):249--290, 1999.

\bibitem{Bachetal1999P}
Volker Bach, J\"{u}rg Fr\"{o}hlich, Israel~Michael Sigal, and Avy Soffer.
\newblock Positive commutators and the spectrum of {P}auli-{F}ierz
  {H}amiltonian of atoms and molecules.
\newblock {\em Comm. Math. Phys.}, 207(3):557--587, 1999.

\bibitem{BetheSalpeter1957Q}
Hans~A. Bethe and Edwin~E. Salpeter.
\newblock {\em {Quantum mechanics of one- and two-electron atoms.}}
\newblock {Berlin-G\"{o}ttingen-Heidelberg: Springer-Verlag. VIII, 368 pp. },
  1957.

  \bibitem{BickelGoodman1966}
William~S. Bickel and Allan~S. Goodman.
\newblock Mean lives of the 2p and 3p levels in atomic hydrogen.
\newblock {\em Physical Review}, 148(2):1--4, August 1970.


\bibitem{Chuppetal1968}
E.~L. Chupp, L.~W. Dotchin, and D.~J. Pegg.
\newblock Radiative mean-life measurements of some atomic-hydrogen excited
  states using beam-foil excitation.
\newblock {\em Physical Review}, 175(1), November 1968.

\bibitem{CohenTannoudjietal1992}
Claude Cohen-Tannoudji, Jacques Dupont-Roc, and Gilbert Grynberg.
\newblock {\em Atom-Photon Interactions -- Basic Processes and Applications}.
\newblock John Wiley and Sons, Inc., 1992.

\bibitem{CohenTannoudjietal2004}
Claude Cohen-Tannoudji, Jacques Dupont-Roc, and Gilbert Grynberg.
\newblock {\em Photons \& Atoms}.
\newblock WILEY-VCH Verlag GmbH\& Co. KGaA, 2004.

\bibitem{Davies1995}
E.~B. Davies.
\newblock The functional calculus.
\newblock {\em J. London Math. Soc. (2)}, 52(1):166--176, 1995.

\bibitem{Ethertonetal1970}
R.~C. Etherton, L.~M. Beyer, W.~E. Maddox, and L.~B. Bridwell.
\newblock Lifetimes of 3p, 4p, and 5p states in atomic hydrogen.
\newblock {\em Physical Review A}, 2(6):2177--2179, December 1970.

\bibitem{Gordon1929}
W.~Gordon.
\newblock {Zur Berechnung der Matrizen beim Wasserstoffatom.}
\newblock {\em Annalen d. Physik}, 5(2):1031--1056, 1929.

\bibitem{HelfferSjostrand1989}
B.~Helffer and J.~Sj\"{o}strand.
\newblock \'{E}quation de {S}chr\"{o}dinger avec champ magn\'etique et \'equation
  de {H}arper.
\newblock In {\em Schr\"{o}dinger operators (S\o nderborg, 1988)}, volume 345 of
  {\em Lecture Notes in Phys.}, pages 118--197. Springer, Berlin, 1989.

\bibitem{Hunziker1990}
Walter Hunziker.
\newblock Resonances, metastable states and exponential decay laws in
  perturbation theory.
\newblock {\em Comm. Math. Phys.}, 132(1):177--188, 1990.

\bibitem{JaksicPillet1995}
V.~Jak{\v{s}}i{\'c} and C.-A. Pillet.
\newblock On a model for quantum friction. {I}. {F}ermi's golden rule and
  dynamics at zero temperature.
\newblock {\em Ann. Inst. H. Poincar\'e Phys. Th\'eor.}, 62(1):47--68, 1995.

\bibitem{King1994}
Christopher King.
\newblock Resonant decay of a two state atom interacting with a massless
  non-relativistic quantised scalar field.
\newblock {\em Comm. Math. Phys.}, 165(3):569--594, 1994.

\end{thebibliography}

\end{document}